\RequirePackage{amsmath}
\RequirePackage{etex}
\documentclass[runningheads]{llncs}

\usepackage[margin=1.7in]{geometry}

\usepackage[utf8]{inputenc}
\usepackage[T1]{fontenc}

\usepackage{fancyhdr}
\usepackage{xcolor}

\usepackage{physics}
\usepackage{amsfonts}
\usepackage{amssymb}
\usepackage{algorithm}
\usepackage{algorithmic}
\usepackage{xfrac}

\newcommand{\bigO}[1]{\mathcal{O}(#1)}

\newcommand{\C}{\mathcal{C}}
\newcommand{\view}[1]{\mathrm{View}_{#1}}

\newcommand{\dist}{\mathsf{D}}
\newcommand{\negl}[1]{\mathrm{negl}(#1)}
\newcommand{\prob}[1]{\mathrm{Pr}[#1]}
\newcommand{\sample}{{\leftarrow}\vcenter{\hbox{\tiny\rmfamily\upshape\$}}}

\newcommand{\covpi}[1]{\Pi^{cov}_{#1}}

\title{Revisiting Deniability in Quantum Key Exchange}
\subtitle{via Covert Communication and Entanglement Distillation}

\titlerunning{Revisiting Deniability in Quantum Key Exchange}

\author{Arash Atashpendar\inst{1}$^{(\star)}$, G. Vamsi Policharla\inst{2}, Peter B. R\o nne\inst{1}, Peter Y.A. Ryan\inst{1}}

\authorrunning{Arash Atashpendar, G. Vamsi Policharla, Peter B. R\o nne, Peter Y. A. Ryan}

\institute{SnT, University of Luxembourg, Luxembourg \email{\{arash.atashpendar,peter.roenne,peter.ryan\}@uni.lu}
\and Department of Physics, Indian Institute of Technology Bombay, India \email{guruvamsi.policharla@iitb.ac.in}}

\begin{document}

	\maketitle
	\thispagestyle{fancy} 

	\begin{abstract}
		We revisit the notion of deniability in quantum key exchange (QKE), a topic that remains largely unexplored. In the only work on this subject by Donald Beaver, it is argued that QKE is not necessarily deniable due to an eavesdropping attack that limits key equivocation.
		We provide more insight into the nature of this attack and how it extends to other constructions such as QKE obtained from uncloneable encryption.
		We then adopt the framework for quantum authenticated key exchange, developed by Mosca et al., and extend it to introduce the notion of coercer-deniable QKE, formalized in terms of the indistinguishability of real and fake coercer views.
		Next, we apply results from a recent work by Arrazola and Scarani on covert quantum communication to establish a connection between covert QKE and deniability. We propose DC-QKE, a simple deniable covert QKE protocol, and prove its deniability via a reduction to the security of covert QKE. Finally, we consider how entanglement distillation can be
		used to enable information-theoretically deniable protocols for QKE and tasks beyond key exchange.
	\end{abstract}

	\section{Introduction}

    Deniability represents a fundamental privacy-related notion in cryptography. The ability to deny a message or an action is a desired property in many contexts such as off-the-record communication, anonymous reporting, whistle-blowing and coercion-resistant secure electronic voting.
    The concept of non-repudiation is closely related to deniability in that the former is aimed at associating specific actions with legitimate parties and thereby preventing them from denying that they have performed a certain task, whereas the latter achieves the opposite property by allowing legitimate parties to deny having performed a particular action. For this reason, deniability is sometimes referred to as \emph{repudiability}.

	The definitions and requirements for deniable exchange can vary depending on the cryptographic task in question, e.g., encryption, authentication or key exchange. Roughly speaking, the common underlying idea for a deniable scheme can be understood as the impossibility for an adversary to produce cryptographic proofs, using only algorithmic evidence, that would allow a third-party, often referred to as a judge, to decide if a particular entity has either taken part in a given exchange or exchanged a certain message, which can be a secret key, a digital signature, or a plaintext message. In the context of key exchange, this can be also formulated in terms of a corrupt party (receiver) proving to a judge that a message can be traced back to the other party \cite{di2006deniable}.

	In the public-key setting, an immediate challenge for achieving deniability is posed by the need for
	remote authentication as it typically gives rise to binding evidence,  e.g., digital signatures, see \cite{di2006deniable,dodis2009composability}. The formal analysis of deniability in classical cryptography can be traced back to the original works of Canetti et al. and Dwork et al. on deniable encryption \cite{canetti1997deniable} and deniable authentication \cite{dwork2004concurrent}, respectively. These led to a series of papers on this topic covering a relatively wide array of applications. Deniable key exchange was first formalized by Di Raimondo et al. in \cite{di2006deniable} using a framework based on the simulation paradigm, which is closely related to that of zero-knowledge proofs.

	\fancyhead[RE]{Revisiting Deniability in Quantum Key Exchange}
	\fancyhead[LO]{A. Atashpendar et al.}

	Despite being a well-known and fundamental concept in classical cryptography, rather surprisingly, deniability has been largely ignored by the quantum cryptography community. To put things into perspective, with the exception of a single paper by Donald Beaver \cite{beaver2002deniability}, and a footnote in \cite{ioannou2011new} commenting on the former, there are no other works that directly tackle deniable QKE.

	In the adversarial setting described in \cite{beaver2002deniability}, it is assumed that the honest parties are approached by the adversary after the termination of a QKE session and demanded to reveal their private randomness, i.e., the raw key bits encoded in their quantum states. It is then claimed that QKE schemes, despite having perfect and unconditional security, are not necessarily deniable due to an eavesdropping attack. In the case of the BB84 protocol, this attack introduces a binding between the parties' inputs and the final key, thus constraining the space of the final secret key such that key equivocation is no longer possible.

	Note that since Beaver's work \cite{beaver2002deniability} appeared a few years before a formal analysis of deniability for key exchange was published, its analysis is partly based on the adversarial model formulated earlier in \cite{canetti1997deniable} for deniable encryption. For this reason, the setting corresponds more closely to scenarios wherein the honest parties try to deceive a coercer by
	presenting fake messages and randomness, e.g., deceiving a coercer who tries to verify a voter's claimed choice using an intercepted ciphertext of a ballot in the context of
	secure e-voting.

	\subsection{Contributions and Structure}

	In Section \ref{sec:coercer-deniable-qke} we revisit the notion of deniability in QKE and provide more insight into the eavesdropping attack aimed at detecting attempts at denial described in \cite{beaver2002deniability}. Having shed light on the nature of this attack, we show that while coercer-deniability can be achieved by uncloneable encryption (UE) \cite{gottesman2002uncloneable}, QKE obtained from UE remains vulnerable to the same attack. We briefly elaborate on the differences between our model and simulation-based deniability \cite{di2006deniable}. To provide a firm foundation, we adopt the framework and security model for quantum authenticated key exchange (Q-AKE) developed by Mosca et al. \cite{mosca2013quantum} and extend them to introduce the notion of coercer-deniable QKE, which we formalize in terms of the indistinguishability of real and fake coercer views.

	We establish a connection between the concept of covert communication and deniability in Section \ref{sec:dc-qke}, which to the best of our knowledge has not been formally considered before. More precisely, we apply results from a recent work by Arrazola and Scarani on obtaining covert quantum communication and covert QKE via noise injection \cite{AS16} to propose DC-QKE, a simple construction for coercer-deniable QKE. We prove the deniability of DC-QKE via a reduction to the security of covert QKE. Compared to the candidate PQECC protocol suggested in \cite{beaver2002deniability} that is claimed to be deniable, our construction does not require quantum computation and falls within the more practical realm of prepare-and-measure protocols.

	Finally, in Section \ref{sec:entanglement-distillation} we consider how quantum entanglement distillation can be used not only to counter eavesdropping attacks, but also to achieve information-theoretic deniability. We conclude by presenting some open questions in Section \ref{sec:open-questions}. It is our hope that this work will rekindle interest, more broadly, in the notion of deniable communication in the quantum setting, a topic that has received very little attention from the quantum cryptography community.

	\subsection{Related Work}

	We focus on some of the most prominent works in the extensive body of work on deniability in classical cryptography. The notion of deniable encryption was considered by Canetti et al. \cite{canetti1997deniable} in a setting where an adversary demands that parties reveal private coins used for generating a ciphertext. This motivated the need for schemes equipped with a faking algorithm that can produce fake randomness with distributions indistinguishable from that of the real encryption.

	In a framework based on the simulation paradigm, Dwork et al. introduced the notion of deniable authentication \cite{dwork2004concurrent}, followed by the work of Di Raimondo et al. on the formalization of deniable key exchange \cite{di2006deniable}. Both works rely on the formalism of zero-knowledge (ZK) proofs, with definitions formalized in terms of a simulator that can produce a simulated view that is indistinguishable from the real one. In a subsequent work, Di Raimondo and Gennaro gave a formal definition of forward deniability \cite{di2009new}, requiring that indistinguishability remain intact even when a (corrupted) party reveals real coins after a session. Among other things, they showed that statistical ZK protocols are forward deniable.

	Pass \cite{pass2003deniability} formally defines the notion of deniable zero-knowledge and presents positive and negative results in the common reference string and random oracle model. In \cite{dodis2009composability}, Dodis et al. establish a link between deniability and ideal authentication and further model a situation in which deniability should hold even when a corrupted party colludes with the adversary during the execution of a protocol. They show an impossibility result in the PKI model if adaptive corruptions are allowed. Cremers and Feltz introduced another variant for key exchange referred to as peer and time deniability \cite{cremers2011one}, while also capturing perfect forward secrecy. More recently, Unger and Goldberg studied deniable authenticated key exchange (DAKE) in the context of secure messaging \cite{unger2015deniable}.

	To the best of our knowledge, the only work related to deniability in QKE is a single paper by Beaver \cite{beaver2002deniability}, in which the author suggests a negative result arguing that existing QKE schemes are not necessarily deniable.

	\section{Preliminaries in Quantum Information and QKE}\label{sec:preliminaries}

	We use the Dirac bra-ket notation and standard terminology from quantum computing. Here we limit ourselves to a description of the most relevant concepts in quantum information theory. More details can be found in standard textbooks \cite{nielsen2002quantum,wilde2013quantum}. For brevity, let $A$ and $B$ denote the honest parties, and $E$ the adversary.

    Given an orthonormal basis formed by $\ket{0}$ and $\ket{1}$ in a two-dimensional complex Hilbert space $\mathcal{H}_2$, let $(+) \equiv \{ \ket{0}, \ket{1} \}$ denote the computational basis and $(\times) \equiv \{ (\sfrac{1}{\sqrt{2}})(\ket{0} + \ket{1}), (\sfrac{1}{\sqrt{2}})(\ket{0} - \ket{1}) \}$ the diagonal basis.

	If the state vector of a composite system cannot be expressed as a tensor product $\ket{\psi_1} \otimes \ket{\psi_2}$, the state of each subsystem cannot be described independently and we say the two qubits are \emph{entangled}. This property is best exemplified by maximally entangled qubits (\emph{ebits}), the so-called \emph{Bell states}
    \begin{align*}
	\ket{\Phi^\pm}_{AB} = \frac{1}{\sqrt{2}}(\ket{00}_{AB} \pm \ket{11}_{AB}) \quad , \quad \ket{\Psi^\pm}_{AB} = \frac{1}{\sqrt{2}}(\ket{01}_{AB} \pm \ket{10}_{AB})
	\end{align*}

	A noisy qubit that cannot be expressed as a linear superposition of pure states is said to be in a \emph{mixed} state, a classical probability distribution of pure states: $\{p_X(x), \ket{\psi_x}\}_{x \in X}$. The \emph{density operator} $\rho$, defined as a weighted sum of projectors, captures both pure and mixed states: $\rho \equiv \sum_{x \in \mathcal{X}}p_X(x) \ket{\psi_x}\bra{\psi_x}$.

	Given a density matrix $\rho_{AB}$ describing the joint state of a system held by $A$ and $B$, the \emph{partial trace} allows us to compute the local state of $A$ (density operator $\rho_A$) if $B$'s system is not accessible to $A$.
	To obtain $\rho_A$ from $\rho_{AB}$ (the reduced state of $\rho_{AB}$ on $A$), we trace out the system $B$: $\rho_A = \mathrm{Tr}_{B}(\rho_{AB})$. As a distance measure, we use the expected fidelity $F(\ket{\psi}, \rho)$ between a pure state $\ket{\psi}$ and a mixed state $\rho$ given by $F(\ket{\psi}, \rho) = \bra{\psi}\rho\ket{\psi}$.

	A crucial distinction between quantum and classical information is captured by the well-known No-Cloning theorem \cite{wootters1982single}, which states that an arbitrary unknown quantum state cannot be copied or cloned perfectly.

	\subsection{Quantum Key Exchange and Uncloneable Encryption}\label{subsec:qke-and-ue}
	QKE allows two parties to establish a common secret key with information-theoretic security using an insecure quantum channel, and a public authenticated classical channel.
	In Protocol \ref{protocol:bb84} we describe the \textbf{BB84} protocol, the most well-known QKE variant due to Bennett and Brassard \cite{bennett1984quantum}. For consistency with related works, we use the well-established formalism based on error-correcting codes, developed by Shor and Preskill \cite{shor2000simple}. Let $C_1[n,k_1]$ and $C_2[n,k_2]$ be two classical linear binary codes encoding $k_1$ and $k_2$ bits in $n$ bits such that $\{0\} \subset C_2 \subset C_1 \subset \mathbf{F}^n_2$ where $\mathbf{F}^n_2$ is the binary vector space on $n$ bits. A mapping of vectors $v \in C_1$ to a set of basis states (codewords) for the Calderbank-Shor-Steane (CSS) \cite{calderbank1996good,steane1996multiple} code subspace is given by: $v \mapsto (\sfrac{1}{\sqrt{|C_2|}})\sum_{w \in C_2}\ket{v+w}$. Due to the irrelevance of phase errors and their decoupling from bit flips in CSS codes, Alice can send $\ket{v}$ along with classical error-correction information $u+v$ where $u,v \in \mathbf{F}^n_2$ and $u \in C_1$, such that Bob can decode to a codeword in $C_1$ from $(v+\epsilon)-(u+v)$ where $\epsilon$ is an error codeword, with the final key being the coset leader of $u + C_2$.

    \begin{algorithm}
	\floatname{algorithm}{Protocol}
	\caption{BB84 for an $n$-bit key with protection against $\delta n$ bit errors}
	\label{protocol:bb84}
	\begin{algorithmic}[1]
		\STATE Alice generates two random bit strings $a,b \in \{0,1\}^{(4+\delta)n}$, encodes $a_i$ into $\ket{\psi_i}$ in basis $(+)$ if $b_i=0$ and in $(\times)$ otherwise, and $\forall i \in [1,|a|]$ sends $\ket{\psi_i}$ to Bob.
		\STATE Bob generates a random bit string $b' \in \{0,1\}^{(4+\delta)n}$ and upon receiving the qubits, measures $\ket{\psi_i}$ in $(+)$ or $(\times)$ according to $b'_i$ to obtain $a'_i$.

		\STATE Alice announces $b$ and Bob discards $a'_i$ where $b_i \neq b'_i$, ending up with at least $2n$ bits with high probability.

		\STATE Alice picks a set $p$ of $2n$ bits at random from $a$, and a set $q$ containing $n$ elements of $p$ chosen as check bits at random. Let $v = p \setminus q$.

		\STATE Alice and Bob compare their check bits and abort if the error exceeds a predefined threshold.

		\STATE Alice announces $u+v$, where $v$ is the string of the remaining non-check bits, and $u$ is a random codeword in $C_1$.

		\STATE Bob subtracts $u+v$ from his code qubits, $v+\epsilon$, and corrects the result, $u+\epsilon$, to a codeword in $C_1$.

		\STATE Alice and Bob use the coset of $u+C_2$ as their final secret key of length $n$.
	\end{algorithmic}
    \end{algorithm}

	\textbf{Uncloneable encryption} (UE) enables transmission of ciphertexts that cannot be perfectly copied and stored for later decoding, by encoding carefully prepared codewords into quantum states, thereby leveraging the No-Cloning theorem. We refer to Gottesman's original work \cite{gottesman2002uncloneable} for a detailed explanation of the sketch in Protocol \ref{protocol:ue}. Alice and Bob agree on a message length $n$, a Message Authentication Code (MAC) of length $s$, an error-correcting code $C_1$ having message length $K$ and codeword length $N$ with distance $2\delta N$ for average error rate $\delta$, and another error-correcting code $C_2$ (for privacy amplification) with message length $K'$ and codeword length $N$ and distance $2(\delta+\eta)N$ to correct more errors than $C_1$, satisfying $C_2^\bot \subset C_1$, where $C_2^\bot$ is the dual code containing all vectors orthogonal to $C_2$. The pre-shared key is broken down into four pieces, all chosen uniformly at random: an authentication key $k \in \{ 0,1\}^s$, a one-time pad $e \in \{0,1\}^{n+s}$, a syndrome $c_1 \in \{0,1\}^{N-K}$, and a basis sequence $b \in \{ 0,1\}^N$.

	\begin{algorithm}
	\floatname{algorithm}{Protocol}
	\caption{Uncloneable Encryption for sending a message $m\in \{0,1\}^n$}
	\label{protocol:ue}
	\begin{algorithmic}[1]
		\STATE Compute $\mathrm{MAC}(m)_k = \mu \in \{0,1\}^s$. Let $x = m || \mu \in \{0,1 \}^{n+s}$.
		\STATE Mask $x$ with the one-time pad $e$ to obtain $y = x \oplus e$.
		\STATE From the coset of $C_1$ given by the syndrome $c_1$, pick a random codeword $z \in \{0,1 \}^N$ that has syndrome bits $y$ w.r.t. $C_2^{\bot}$, where $C_2^\bot \subset C_1$.
		\STATE For $i \in [1, N]$ encode ciphertext bit $z_i$ in the basis $(+)$ if $b_i = 0$ and in the basis $(\times)$ if $b_i = 1$. The resulting state $\ket{\psi_i}$ is sent to Bob.
	\end{algorithmic}
	To perform decryption:
	\begin{algorithmic}[1]
		\STATE For $i \in [1, N]$, measure $\ket{\psi'_i}$ according to $b_i$, to obtain $z'_i \in \{0,1\}^N$.
		\STATE Perform error-correction on $z'$ using code $C_1$ and evaluate the parity checks of $C_2/C_1^{\bot}$ for privacy amplification to get an $(n+s)$-bit string $y'$.
		\STATE Invert the OTP step to obtain $x' = y' \oplus e$.
		\STATE Parse $x'$ as the concatenation $m' || \mu'$ and use $k$ to verify if $\mathrm{MAC}(m')_k = \mu'$.
	\end{algorithmic}
    \end{algorithm}

	\paragraph{QKE from UE.} It is known \cite{gottesman2002uncloneable} that any quantum authentication (QA) scheme can be used as a secure UE scheme, which can in turn be used to obtain QKE, with less interaction and more efficient error detection. We give a brief description of how QKE can be obtained from UE in Protocol \ref{protocol:qke-from-ue}.

	\begin{algorithm}
	\floatname{algorithm}{Protocol}
	\caption{Obtaining QKE from Uncloneable Encryption}
	\label{protocol:qke-from-ue}
	\begin{algorithmic}[1]
	    \STATE Alice generates random strings $k$ and $x$, and sends $x$ to Bob via UE, keyed with $k$.
	    \STATE Bob announces that he has received the message, and then Alice announces $k$.
	    \STATE Bob decodes the classical message $x$, and upon MAC verification, if the message is valid, he announces this to Alice and they will use $x$ as their secret key.
	\end{algorithmic}
    \end{algorithm}

    \section{Coercer-Deniable Quantum Key Exchange}\label{sec:coercer-deniable-qke}

	Following the setting in \cite{beaver2002deniability}, in which it is implicitly assumed that the adversary has established a binding between the participants' identities and a given QKE session, we introduce the notion of coercer-deniability for QKE. This makes it possible to consider an adversarial setting similar to that of deniable encryption \cite{canetti1997deniable} and expect that the parties might be coerced into revealing their private coins after the termination of a session, in which case they would have to produce fake randomness such that the resulting transcript and the claimed values remain consistent with the adversary's observations.

	Beaver's analysis \cite{beaver2002deniability} is briefly addressed in a footnote in a paper by Ioannou and Mosca \cite{ioannou2011new} and the issue is brushed aside based on the argument that the parties do not have to keep records of their raw key bits. It is argued that for deniability to be satisfied, it is sufficient that the adversary cannot provide binding evidence that attributes a particular key to the classical communication as their measurements on the quantum channel do not constitute a publicly verifiable proof. However, counter-arguments for this view were already raised in the motivations for deniable encryption \cite{canetti1997deniable} in terms of secure erasure being difficult and unreliable, and that erasing cannot be externally verified. Moreover, it is also argued that if one were to make the physical security assumption that random choices made for encryption are physically unavailable, the deniability problem would disappear. We refer to \cite{canetti1997deniable} and references therein for more details.

	Bindings, or lack thereof, lie at the core of deniability. Although we leave a formal comparison of our model with the one formulated in the simulation paradigm \cite{di2006deniable} as future work, a notable difference can be expressed in terms of the inputs presented to the adversary. In the simulation paradigm, deniability is modelled only according to the simulatability of the legal transcript that the adversary or a corrupt party produces naturally via a session with a party as evidence for the judge, whereas for coercer-deniability, the adversary additionally demands that the honest parties reveal their private randomness.

	Finally, note that viewing deniability in terms of ``convincing'' the adversary is bound to be problematic and indeed a source of debate in the cryptographic research community as the adversary may never be convinced given their knowledge of the existence of faking algorithms.
	Hence, deniability is formulated in terms of the indistinguishability of views (or their simulatability \cite{di2006deniable}) such that a judge would have no reason to believe a given transcript provided by the adversary establishes a binding as it could have been forged or simulated.

	\subsection{Defeating Deniability in QKE via Eavesdropping in a Nutshell}\label{subsec:state-injection-attack}

	We briefly review the eavesdropping attack described in \cite{beaver2002deniability} and provide further insight. Suppose Alice sends qubit $\ket{\psi}^{m,b}$ to Bob, which encodes a single-bit message $m$ prepared in a basis determined by $b \in \{+, \times\}$. Let $\Phi(E, m)$ denote the state obtained after sending $\ket{\psi}^{m,b}$, relayed and possibly modified by an adversary $E$. Moreover, let $\rho(E, m)$ denote the view presented to the judge, obtained by tracing over inaccessible systems. Now for a qubit measured correctly by Eve, if a party tries to deny by pretending to have sent $\sigma_1 = \rho(E, 1)$ instead of $\sigma_2  = \rho(E, 0)$, e.g., by using some local transformation $U_{neg}$ to simply negate a given qubit, then $F(\sigma_1, \sigma_2) = 0$, where $F$ denotes the fidelity between $\sigma_1$ and $\sigma_2$. Thus, the judge can successfully detect this attempt at denial.

	This attack can be mounted successfully with non-negligible probability without causing the session to abort: Assume that $N$ qubits will be transmitted in a BB84 session and that the tolerable error rate is $\frac{\eta}{N}$, where clearly $\eta \sim N$. Eve measures each qubit with probability $\frac{\eta}{N}$ (choosing a basis at random) and passes on the remaining ones to Bob undisturbed, i.e., she plants a number of decoy states proportional to the tolerated error threshold. On average, $\frac{\eta}{2}$ measurements will come from matching bases, which can be used by Eve to detect attempts at denial, if Alice claims to have measured a different encoding. After discarding half the qubits in the sifting phase, this ratio will remain unchanged. Now Alice and/or Bob must flip at least one bit in order to deny without knowledge of where the decoy states lie in the transmitted sequence, thus getting caught with probability $\frac{\eta}{2N}$ upon flipping a bit at random.

	\subsection{On the Coercer-Deniability of Uncloneable Encryption}

	The vulnerability described in Section \ref{subsec:state-injection-attack} is made possible by an eavesdropping attack that induces a binding in the key coming from a BB84 session. Uncloneable encryption remains immune to this attack because the quantum encoding is done for an already one-time padded classical input. More precisely, a binding established at the level of quantum states can still be perfectly denied because the actual raw information bits $m$ are not directly encoded into the sequence of qubits, instead the concatenation of $m$ and the corresponding authentication tag $\mu = \mathrm{MAC}_k(m)$, i.e., $x=m||\mu$, is masked with a one-time pad $e$ to obtain $y = x \oplus e$, which is then mapped onto a codeword $z$ that is encoded into quantum states. For this reason, in the context of coercer-deniability, regardless of a binding established on $z$ by the adversary, Alice can still deny to another input message in that she can pick a different input $x'=m'||\mu'$ to compute a fake pad $e' = y \oplus x'$, so that upon revealing $e'$ to Eve, she will simply decode $y \oplus e' = x'$, as intended.

	However, note that a prepare-and-measure QKE obtained from UE still remains vulnerable to the same eavesdropping attack due to the fact that we can no longer make use of the deniability of the one-time pad in UE such that the bindings induced by Eve constrain the choice of the underlying codewords.

	\subsection{Security Model}\label{subsec:security-model}

	We adopt the framework for quantum AKEs developed by Mosca et al. \cite{mosca2013quantum}. Due to space constraints, we mainly focus on our proposed extensions. \textbf{Parties}, including the adversary, are modelled as a pair of classical and quantum Turing machines (TM) that execute a series of interactive computations and exchange messages with each other through classical and quantum channels, collectively referred to as a \textbf{protocol}. An execution of a protocol is referred to as a \textbf{session}, identified with a unique session identifier.
	An ongoing session is called an \emph{active} session, and upon completion, it either outputs an error term $\bot$ in case of an abort, or it outputs a tuple $(sk, pid, \mathbf{v}, \mathbf{u})$ in case of a successful termination. The tuple consists of a session key $sk$, a party identifier $pid$ and two vectors $\mathbf{u}$ and $\mathbf{v}$ that model public values and secret terms, respectively.

    We adopt an extended version of the \textbf{adversarial model} described in \cite{mosca2013quantum}, to account for coercer-deniability. Let $E$ be an efficient, i.e. (quantum) polynomial time, adversary with classical and quantum runtime bounds $t_c(k)$ and $t_q(k)$, and quantum memory bound $m_q(k)$, where bounds can be unlimited. Following standard assumptions, the adversary controls all communication between parties and carries the messages exchanged between them. We consider an authenticated classical channel and do not impose any special restrictions otherwise. Additionally, the adversary is allowed to approach either the sender or the receiver after the termination of a session and request access to a subset $
    \vec{r} \subseteq \vec{v}$ of the private randomness used by the parties for a given session, i.e. set of values to be faked.

	Security notions can be formulated in terms of \textbf{security experiments} in which the adversary interacts with the parties via a set of well-defined \textbf{queries}. These queries typically involve sending messages to an active session or initiating one, corrupting a party, learning their long-term secret key, revealing the ephemeral keys of an incomplete session, obtaining the computed session key for a given session, and a \textbf{test-session($id$)} query capturing the winning condition of the game that can be invoked only for a \emph{fresh} session. Revealing secret values to the adversary is modeled via \textbf{partnering}. The notion of \emph{freshness} captures the idea of excluding cases that would allow the adversary to trivially win the security experiment. This is done by imposing minimal restrictions on the set of queries the adversary can invoke for a given session such that there exist protocols that can still satisfy the definition of session-key security.
	A session remains fresh as long as at least one element in $\vec{u}$ and $\vec{v}$ remains secret, see \cite{mosca2013quantum} for more details.

	The \textbf{transcript} of a protocol consists of all publicly exchanged messages between the parties during a run or session of the protocol.
	The definition of ``views'' and ``outputs'' given in \cite{beaver2002deniability} coincides with that of transcripts in \cite{di2006deniable} in the sense that it allows us to model a transcript that can be obtained from observations made on the quantum channel. The \emph{view} of a party $P$ consists of their state in $\mathcal{H}_P$ along with any classical strings they produce or observe. More generally, for a two-party protocol, captured by the global density matrix $\rho_{AB}$ for the systems of $A$ and $B$, the individual system $A$ corresponds to a partial trace that yields a reduced density matrix, i.e., $\rho_A = \mathrm{Tr}_B(\rho_{AB})$, with a similar approach for any additional couplings.

	\subsection{Coercer-Deniable QKE via View Indistinguishability}

	We use the security model in Section \ref{subsec:security-model} to introduce the notion of coercer-deniable QKE, formalized via the indistinguishability of real and fake views. Note that in this work we do not account for forward deniability and forward secrecy.

	\paragraph{Coercer-Deniability Security Experiment.}\label{sec-exp:coercer-deniable-qke}
	Let $\mathrm{CoercerDenQKE}^{\Pi}_{E, \C}(\kappa)$ denote this experiment and $Q$ the same set of queries available to the adversary in a security game for session-key security, as described in Section \ref{subsec:security-model}, and \cite{mosca2013quantum}. Clearly, in addition to deniability, it is vital that the security of the session key remains intact as well. For this reason, we simply extend the requirements of the security game for a session-key secure KE by having the challenger $\C$ provide an additional piece of information to the adversary $E$ when the latter calls the \textbf{test-session()} query. This means that the definition of a fresh session remains the same as the one given in \cite{mosca2013quantum}. $E$ invokes queries from $Q \setminus \{\text{{\bfseries test-session()}}\}$ until $E$ issues \textbf{test-session()} to a fresh session of their choice. $\C$ decides on a random bit $b$ and if $b=0$, $\C$ provides $E$ with the real session key $k$ and the real vector of private randomness $\vec{r}$, and if $b=1$, with a random (fake) key $k'$ and a random (fake) vector of private randomness $\vec{r}'$.
	Finally, $E$ guesses an output $b'$ and wins the game if $b = b'$. The experiment returns 1 if $E$ succeeds, and 0 otherwise. Let $Adv_{E}^{\Pi}(\kappa) = |\prob{b = b'} - \sfrac{1}{2}|$ denote the winning advantage of $E$.

	\begin{definition}[Coercer-Deniable QKE]\label{def:coercer-deniable-qke}
	For adversary $E$, let there be an efficient distinguisher $\dist_E$ on security parameter $\kappa$. We say that $\Pi_{\vec{r}}$ is a coercer-deniable QKE protocol if, for any adversary $E$, transcript $\vec{t}$, and for any $k, k'$, and a vector of private random inputs $\vec{r} = (r_1, \ldots, r_{\ell})$, there exists a denial/faking program $\mathcal{F}_{A,B}$ that running on $(k, k', \vec{t}, \vec{r})$ produces $\vec{r}' = (r'_1, \ldots, r'_{\ell})$ such that the following conditions hold:
	\begin{itemize}
	    \item $\Pi$ is a secure QKE protocol.
	    \item The adversary $E$ cannot do better than making a random guess for winning the coercer-deniability security experiment, i.e., $Adv_{E}^{\Pi}(\kappa) \le \mathrm{negl}(\kappa)$
	    \[
	    \mathrm{Pr}[\mathrm{CoercerDenQKE}^{\Pi}_{E,\C}(\kappa) = 1] \le \frac{1}{2} + \mathrm{negl}(\kappa)
	    \]
	\end{itemize}
	Equivalently, we require that for all efficient distinguisher $\dist_E$
	\[
	|\prob{\dist_E(\view{Real}(k, \vec{t}, \vec{r})) = 1} - \prob{\dist_E(\view{Fake}(k', \vec{t}, \vec{r'})) = 1}| \le \mathrm{negl}(\kappa),
	\]
	where the transcript $\vec{t}=(\vec{c}, \rho_E(k))$ is a tuple consisting of a vector $\vec{c}$, containing classical message exchanges of a session, along with the local view of the adversary w.r.t. the quantum channel obtained by tracing over inaccessible systems (see Section \ref{subsec:security-model}).
	\end{definition}

	A function $f: \mathbb{N} \rightarrow \mathbb{R}$ is negligible if for any constant $k$, there exists a $N_k$ such that $\forall  N \ge N_k$, we have $f(N) < N^{-k}$. In other words, it approaches zero faster than any polynomial in the asymptotic limit.

	\begin{remark}\label{remark:randomness-compromise}
	We introduced a vector of private random inputs $\vec{r}$ to avoid being restricted to a specific set of ``fake coins'' in a coercer-deniable setting such as the raw key bits in BB84 as used in Beaver's analysis. This allows us to include other private inputs as part of the transcript that need to be forged by the denying parties without having to provide a new security model for each variant. Indeed, in \cite{mosca2013quantum}, Mosca et al. consider the security of QKE in case various secret values are compromised before or after a session. This means that these values can, in principle, be included in the set of random coins that might have to be revealed to the adversary and it should therefore be possible to generate fake alternatives using a faking algorithm.
	\end{remark}

	\section{Deniable QKE via Covert Quantum Communication}\label{sec:dc-qke}

	We establish a connection between covert communication and deniability by providing a simple construction for coercer-deniable QKE using covert QKE. We then show that deniability is reduced to the covertness property, meaning that deniable QKE can be performed as long as covert QKE is not broken by the adversary, formalized via the security reduction given in Theorem \ref{thm:den-covert-reduction}.

	Covert communication becomes relevant when parties wish to keep the very act of communicating secret or hidden from a malicious warden. This can be motivated by various requirements such as the need for hiding one's communication with a particular entity when this act alone can be incriminating. While encryption can make it impossible for the adversary to access the contents of a message, it would not prevent them from detecting exchanges over a channel under their observation. Bash et al. \cite{bash2015hiding,sheikholeslami2016covert} established a square-root law for covert communication in the presence of an unbounded quantum adversary stating that $\bigO{\sqrt{n}}$ covert bits can be exchanged over $n$ channel uses. Recently, Arrazola and Scarani \cite{AS16} extended covert communication to the quantum regime for transmitting qubits covertly. Covert quantum communication consists of two parties exchanging a sequence of qubits such that an adversary trying to detect this cannot succeed by doing better than making a random guess, i.e., $P_d \le \frac{1}{2} + \epsilon$ for sufficiently small $\epsilon > 0$, where $P_d$ denotes the probability of detection and $\epsilon$ the detection bias.

	\subsection{Covert Quantum Key Exchange}

	Since covert communication requires pre-shared secret randomness, a natural question to ask is whether QKE can be done covertly. This was also addressed in \cite{AS16} and it was shown that covert QKE with unconditional security for the covertness property is impossible because the amount of key consumption is greater than the amount produced. However, a hybrid approach involving pseudo-random number generators (PRNG) was proposed to achieve covert QKE with a positive key rate such that the resulting secret key remains information-theoretically secure, while the covertness of QKE is shown to be at least as strong as the security of the PRNG. The PRNG is used to expand a truly random pre-shared key into an exponentially larger pseudo-random output, which is then used to determine the time-bins for sending signals in covert QKE.

	\paragraph{Covert QKE Security Experiment.}\label{sec-exp:covert-qke} Let $\mathrm{CovertQKE}^{\Pi^{cov}}_{E,\C}(\kappa)$ denote the security experiment. The main property of covert QKE, denoted by $\covpi{}$, can be expressed as a game played by the adversary $E$ against a challenger $\C$ who decides on a random bit $b$ and if $b=0$, $\C$ runs $\covpi{}$, otherwise (if $b=1$), $\C$ does not run $\covpi{}$. Finally, $E$ guesses a random bit $b'$ and wins the game if $b=b'$. The experiment outputs 1 if $E$ succeeds, and 0 otherwise.
	The winning advantage of $E$ is given by $Adv_{E}^{\Pi^{cov}}(\kappa) = |\prob{b = b'} - \sfrac{1}{2}|$ and we want that $Adv^{\Pi^{cov}}_{E}(\kappa) \le \negl{\kappa}$.

	\begin{definition}\label{def:covert-QKE}
	Let $G: \{0,1\}^s \rightarrow \{0,1\}^{g(s)}$ be a $(\tau,\epsilon)$-PRNG secure against all efficient distinguishers $\dist$ running in time at most $\tau$ with success probability at most $\epsilon$, where $\forall s: g(s) > s$. A QKE protocol $\covpi{G}$ is considered to be covert if the following holds for any efficient adversary $E$:
		\begin{itemize}
		    \item $\covpi{G}$ is a secure QKE protocol.
			\item The probability that $E$ guesses the bit $b$ correctly ($b' = b$), i.e., $E$ manages to distinguish between Alice and Bob running $\covpi{G}$ or not, is no more than $\frac{1}{2}$ plus a negligible function in the security parameter $\kappa$, i.e.,
			\[
			\prob{\mathrm{CovertQKE}^{\Pi^{cov}}_{E, \C}(\kappa) = 1} \le \frac{1}{2} + \negl{\kappa}
			\]
		\end{itemize}
	\end{definition}
	\begin{theorem}(Sourced from \cite{AS16})\label{thm:covert-QKE}
	The secret key obtained from the covert QKE protocol $\covpi{G}$ is informational-theoretically secure and the covertness of $\covpi{G}$ is as secure as the underlying PRNG.
    \end{theorem}

    \subsection{Deniable Covert Quantum Key Exchange (DC-QKE)}

    We are now in a position to describe DC-QKE, a simple construction shown in Protocol \ref{protocol:dc-qke}, which preserves unconditional security for the final secret key, while its deniability is as secure as the underlying PRNG used in $\covpi{\vec{r},G}$. In terms of the Security Experiment \ref{sec-exp:coercer-deniable-qke}, $\covpi{\vec{r},G}$ is run to establish a real key $k$, while non-covert QKE $\Pi_{\vec{r}'}$ is used to produce a fake key $k'$ aimed at achieving deniability, where $\vec{r}$ and $\vec{r}'$ are the respective vectors of real and fake private inputs.

    Operationally, consider a setting wherein the parties suspect in advance that they might be coerced into revealing their private coins for a given run: their joint strategy consists of running both components in Protocol \ref{protocol:dc-qke} and claiming to have employed $\Pi_{\vec{r}'}$ to establish the fake key $k'$ using the fake private randomness $\vec{r}'$ (e.g. raw key bits in BB84) and provide these as input to the adversary upon termination of a session. Thus, for Eve to be able to produce a proof showing that the revealed values are fake, she would have to break the security of covert QKE to detect the presence of $\covpi{\vec{r},G}$, as shown in Theorem \ref{thm:den-covert-reduction}. Moreover, note that covert communication can be used for dynamically agreeing on a joint strategy for denial, further highlighting its relevance for deniability.

    \begin{algorithm}
	\floatname{algorithm}{Protocol}
	\caption{DC-QKE for an $n$-bit key}
	\label{protocol:dc-qke}
	\begin{algorithmic}[1]
        \STATE \textbf{RandGen:} Let $\vec{r} = (r_1, \ldots, r_{\ell})$ be the vector of private random inputs, where $r_i \sample \{0,1\}^{|r_i|}$.
        \STATE \textbf{KeyGen:} Run $\Pi^{cov}_{\vec{r},G}$ to establish a random secret key $k \in \{0,1\}^{n}$.
	\end{algorithmic}
	Non-covert faking component $\mathcal{F}_{A,B}$:
	\begin{algorithmic}[1]
        \STATE \textbf{FakeRandGen:} Let $\vec{r}' = (r'_1, \ldots, r'_{\ell})$ be the vector of fake private random inputs, where $r'_i \sample \{0,1\}^{|r'_i|}$.
        \STATE \textbf{FakeKeyGen:} Run $\Pi_{\vec{r'}}$ to establish a separate fake key $k' \in \{0,1\}^n$.
	\end{algorithmic}
    \end{algorithm}

    \begin{remark}
	The original analysis in \cite{beaver2002deniability} describes an attack based solely on revealing fake raw key bits that may be inconsistent with the adversary's observations. An advantage of DC-QKE in this regard is that Alice's strategy for achieving coercer-deniability consists of revealing all the secret values of the non-covert QKE $\Pi_{\vec{r}'}$ honestly.
	This allows her to cover the full range of private randomness that could be considered in different variants of deniability as discussed in Remark \ref{remark:randomness-compromise}. A potential drawback is the extra cost induced by $\mathcal{F}_{A,B}$, which could, in principle, be mitigated using a less interactive solution such as QKE via UE.
	\end{remark}

	\begin{remark}
	If the classical channel is authenticated by an information-theoretically secure algorithm, the minimal entropy overhead in terms of pre-shared key (logarithmic in the input size) for $\Pi$ can be generated by $\covpi{\vec{r}}$.
    \end{remark}

	\begin{example}
	In the case of encryption, $A$ can send $c = m \oplus k$ over a covert channel to $B$, while for denying to $m'$, she can send $c' = m' \oplus k'$ over a non-covert channel. Alternatively, she can transmit a single ciphertext over a non-covert channel such that it can be opened to two different messages. To do so, given $c = m \oplus k$, Alice computes $k' = m' \oplus c = m' \oplus m \oplus k$, and she can then either encode $k'$ as a codeword, as described in Section \ref{subsec:qke-and-ue}, and run $\Pi_{\vec{r}'}$ via uncloneable encryption, thus allowing her to reveal the entire transcript to Eve honestly, or she can agree with Bob on a suitable privacy amplification (PA) function (with PA being many-to-one) as part of their denying program in order to obtain $k'$.
	\end{example}

	\begin{theorem}\label{thm:den-covert-reduction}
	If $\covpi{\vec{r},G}$ is a covert QKE protocol, then DC-QKE given in Protocol \ref{protocol:dc-qke} is a coercer-deniable QKE protocol that satisfies Definition \ref{def:coercer-deniable-qke}.
	\end{theorem}
	\begin{proof}
	The main idea consists of showing that breaking the deniability property of DC-QKE amounts to breaking the security of covert QKE, such that coercer-deniability follows from the contrapositive of this implication, i.e., if there exists no efficient algorithm for compromising the security of covert QKE, then there exists no efficient algorithm for breaking the deniability of DC-QKE. We formalize this via a reduction, sketched as follows. Let $w' = \view{Fake}(k', \vec{t}_E, \vec{r}')$ and $w = \view{Real}(k, \vec{t}_E, \vec{r})$ denote the two views. Flip a coin $b$ for an attempt at denial: if $b=0$, then  $\vec{t}_E=(\vec{t}',\varnothing)$, else ($b=1$), $\vec{t}_E=
	(\vec{t}', \vec{t}^{cov})$, where $\vec{t}^{cov}$ and $\vec{t}'$ denote the transcripts of covert and non-covert exchanges from $\covpi{\vec{r},G}$ and $\Pi_{\vec{r}'}$.
	Now if DC-QKE is constructed from $\Pi^{cov}$, then given an efficient adversary $E$ that can distinguish $w$ from $w'$ with probability $p_1$, we can use $E$ to construct an efficient distinguisher $\dist$ to break the security of covert QKE with probability $p_2$ such that $p_1 \le p_2$. Indeed, given an instance of a DC-QKE security game, we construct a distinguisher $\dist$ that uses $E$ on input $w$ and $w'$, with the goal to win the game described in the Security Experiment \ref{sec-exp:coercer-deniable-qke}. The distinguisher $\dist$ would simply run $E$ (with negligible overhead) and observe whether $E$ succeeds at distinguishing $w$ from $w'$. Since the only element that is not sampled uniformly at random is in $\vec{t}^{cov}$ containing exchanges from the covert channel, which relies on a PRNG, the only way $E$ can distinguish $w$ from $w'$ is if she can distinguish $(\vec{t}', \vec{t}^{cov})$ from $(\vec{t}', \varnothing)$. If $E$ succeeds, then $\dist$ guesses that a covert QKE session has taken place, thereby winning the Security Experiment \ref{sec-exp:covert-qke} for covert QKE. \qed
	\end{proof}

	\section{Deniability via Entanglement Distillation}\label{sec:entanglement-distillation}

	Here we consider the possibility of achieving information-theoretic deniability via entanglement distillation (ED). In its most general form, ED allows two parties to distill maximally entangled pure states (\emph{ebits}) from an arbitrary sequence of entangled states at some positive rate using local operations and classical communication (LOCC), i.e. to move from $\ket{\Phi_\theta}_{AB} \equiv cos(\theta)\ket{00}_{AB} + sin(\theta)\ket{11}_{AB}$ to $\ket{\Phi^+}_{AB} = \frac{1}{\sqrt{2}} (\ket{00}_{AB} + \ket{11}_{AB})$, where $0 < \theta < \pi/2$.

	In the noiseless model, $n$ independent identically distributed (i.i.d.) copies of the same partially entangled state $\rho$ can be converted into $\approx nH(\rho)$ Bell pairs in the limit $n \rightarrow \infty$, i.e., from $\rho_{AB}^{
	\otimes n}$ to $\ket{\Phi^+}_{AB}^{\otimes nH(\rho)}$, where $H(\rho) = -\mathrm{Tr}(\rho \mathrm{ln} \rho)$ denotes the von Neumann entropy of entanglement. If the parties start out with pure states, local operations alone will suffice for distillation \cite{bennett1996concentrating,bennett1996purification}, otherwise the same task can be achieved via forward classical communication (one-way LOCC), as shown by the Devetak-Winter theorem \cite{devetak2005distillation}, to distill ebits from many copies of some bipartite entangled state. See also the early work of Bennett et al. \cite{bennett1996mixed} on mixed state ED. Buscemi and Datta \cite{buscemi2010distilling} relax the i.i.d. assumption and provide a general formula for the optimal rate at which ebits can be distilled from a noisy and arbitrary source of entanglement via one-way and two-way LOCC.

	Intuitively, the eavesdropping attack described in \cite{beaver2002deniability} and further detailed in Section \ref{subsec:state-injection-attack}, is enabled by the presence of noise in the channel as well as the fact that Bob cannot distinguish states sent by Alice from those prepared by Eve. As a result, attempting to deny to a different bit value encoded in a given quantum state - without knowing if this is a decoy state prepared by Eve - allows the adversary to detect such an attempt with non-negligible probability.

	In terms of deniability, the intuition behind this idea is that while Alice and Bob may not be able to know which states have been prepared by Eve, they can instead remove her ``check'' decoy states from their set of shared entangled pairs by decoupling her system from theirs. Once they are in possession of maximally entangled states, they will have effectively factored out Eve's state such that the global system is given by the pure tensor product space $\ket{\Psi^+}_{AB} \otimes \ket{\phi}_E$. Thus the pure bipartite joint system between Alice and Bob cannot be correlated with any system under Eve's control, thereby foiling her cross-checking strategy. The singlet states can then be used to perform QKE via quantum teleportation \cite{bennett1993teleporting}.

	\subsection{Deniable QKE via Entanglement Distillation and Teleportation}

	We now argue why performing randomness distillation at the quantum level, thus requiring quantum computation, plays an important role w.r.t. deniability.
	The subtleties alluded to in \cite{beaver2002deniability} arise from the fact that randomness distillation is performed in the classical post-processing step. This allows Eve to leverage her tampering in that she can verify the parties' claims against her decoy states. However, this attack can be countered by removing Eve's knowledge before the classical exchanges begin. Most security proofs of QKE \cite{lo1999unconditional,shor2000simple,mayers2001unconditional} are based on a reduction to an entanglement-based variant, such that the fidelity of Alice and Bob's final state with $\ket{\Psi^+}^{\otimes m}$ is shown to be exponentially close to 1. Moreover, secret key distillation techniques involving ED and quantum teleportation \cite{bennett1996purification,devetak2005distillation} can be used to faithfully transfer qubits from $A$ to $B$ by consuming ebits. To illustrate the relevance of distillation for deniability in QKE, consider the generalized template shown in Protocol \ref{protocol:distillation-qke}, based on these well-known techniques.
	\begin{algorithm}
	\floatname{algorithm}{Protocol}
	\caption{{\small Template for deniable QKE via entanglement distillation and teleportation}}
	\label{protocol:distillation-qke}
	\begin{algorithmic}[1]
	    \STATE $A$ and $B$ share $n$ noisy entangled pairs (assume i.i.d. states for simplicity).
	    \STATE They perform entanglement distillation to convert them into a state $\rho$ such that $F(\ket{\Psi^+}^{\otimes m},\rho)$ is arbitrarily close to 1 where $m < n$.
        \STATE Perform verification to make sure they share $m$ maximally entangled states $\ket{\Psi^+}^{\otimes m}$, and abort otherwise.
        \STATE $A$ prepares $m$ qubits (e.g. BB84 states) and performs quantum teleportation to send them to $B$ at the cost of consuming $m$ ebits and exchanging $2m$ classical bits.
        \STATE $A$ and $B$ proceed with standard classical distillation techniques to agree on a key based on their measurements.
	\end{algorithmic}
    \end{algorithm}

	By performing ED, Alice and Bob make sure that the resulting state cannot be correlated with anything else due to the monogamy of entanglement (see e.g. \cite{koashi2004monogamy,streltsov2012general}), thus factoring out Eve's system.
	The parties can open their records for steps $(2)$ and $(3)$ honestly, and open to arbitrary classical inputs for steps $(3), (4)$ and $(5)$: deniability follows from decoupling Eve's system, meaning that she is faced with a reduced density matrix on a pure bipartite maximally entangled state, i.e., a maximally mixed state $\rho_E = \mathbb{I}/2$, thus obtaining key equivocation.

	In terms of the hierarchy of entanglement-based constructions mentioned in \cite{beaver2002deniability}, this approach mainly constitutes a generalization of such schemes. It should therefore be viewed more as a step towards a theoretical characterization of entanglement-based schemes for achieving information-theoretic deniability. Due to lack of space, we omit a discussion of how techniques from device-independent cryptography can deal with maliciously prepared initial states.

	Going beyond QKE, note that quantum teleportation allows the transfer of an \emph{unknown} quantum state, meaning that even the sender would be oblivious as to what state is sent. Moreover, ebits can enable uniquely quantum tasks such as \emph{traceless exchange} in the context of quantum anonymous transmission \cite{christandl2005quantum}, to achieve \emph{incoercible} protocols that allow parties to deny to any random input.

	\section{Open Questions and Directions for Future Research}\label{sec:open-questions}

	Studying the deniability of public-key authenticated QKE both in our model and in the simulation paradigm, and the existence of an equivalence relation between our indistinguishability-based definition and a simulation-based one would be a natural continuation of this work.
	Other lines of inquiry include forward deniability, deniable QKE in conjunction with forward secrecy, deniability using covert communication in stronger adversarial models, a further analysis of the relation between the impossibility of unconditional quantum bit commitment and deniability mentioned in \cite{beaver2002deniability}, and deniable QKE via uncloneable encryption.
	Finally, gaining a better understanding of entanglement distillation w.r.t. potential pitfalls in various adversarial settings and proposing concrete deniable protocols for QKE and other tasks beyond key exchange represent further research avenues.

	\subsection*{Acknowledgments}

	We thank Mark M. Wilde and Ignatius William Primaatmaja for their comments. This work was supported by a grant (Q-CoDe) from the Luxembourg FNR.

    \bibliographystyle{splncs04}
	\bibliography{references}

\begin{thebibliography}{10}
\providecommand{\url}[1]{\texttt{#1}}
\providecommand{\urlprefix}{URL }
\providecommand{\doi}[1]{https://doi.org/#1}

\bibitem{AS16}
Arrazola, J.M., Scarani, V.: Covert quantum communication. Physical review
  letters  \textbf{117}(25),  250503 (2016)

\bibitem{bash2015hiding}
Bash, B.A., Goeckel, D., Towsley, D., Guha, S.: Hiding information in noise:
  Fundamental limits of covert wireless communication. IEEE Communications
  Magazine  \textbf{53}(12),  26--31 (2015)

\bibitem{beaver2002deniability}
Beaver, D.: On deniability in quantum key exchange. In: Knudsen, L.R. (ed.)
  Advances in Cryptology --- EUROCRYPT 2002. pp. 352--367. Springer Berlin
  Heidelberg, Berlin, Heidelberg (2002)

\bibitem{bennett1996concentrating}
Bennett, C.H., Bernstein, H.J., Popescu, S., Schumacher, B.: Concentrating
  partial entanglement by local operations. Physical Review A  \textbf{53}(4),
  ~2046 (1996)

\bibitem{bennett1984quantum}
Bennett, C.H., Brassard, G.: Quantum cryptography: public key distribution and
  coin tossing int. In: Conf. on Computers, Systems and Signal Processing
  (Bangalore, India, Dec. 1984). pp. 175--9 (1984)

\bibitem{bennett1993teleporting}
Bennett, C.H., Brassard, G., Cr{\'e}peau, C., Jozsa, R., Peres, A., Wootters,
  W.K.: Teleporting an unknown quantum state via dual classical and
  einstein-podolsky-rosen channels. Physical review letters  \textbf{70}(13),
  ~1895 (1993)

\bibitem{bennett1996purification}
Bennett, C.H., Brassard, G., Popescu, S., Schumacher, B., Smolin, J.A.,
  Wootters, W.K.: Purification of noisy entanglement and faithful teleportation
  via noisy channels. Physical review letters  \textbf{76}(5), ~722 (1996)

\bibitem{bennett1996mixed}
Bennett, C.H., DiVincenzo, D.P., Smolin, J.A., Wootters, W.K.: Mixed-state
  entanglement and quantum error correction. Physical Review A  \textbf{54}(5),
  ~3824 (1996)

\bibitem{buscemi2010distilling}
Buscemi, F., Datta, N.: Distilling entanglement from arbitrary resources.
  Journal of Mathematical Physics  \textbf{51}(10),  102201 (2010)

\bibitem{calderbank1996good}
Calderbank, A.R., Shor, P.W.: Good quantum error-correcting codes exist.
  Physical Review A  \textbf{54}(2), ~1098 (1996)

\bibitem{canetti1997deniable}
Canetti, R., Dwork, C., Naor, M., Ostrovsky, R.: Deniable encryption. In:
  Annual International Cryptology Conference. pp. 90--104. Springer (1997)

\bibitem{christandl2005quantum}
Christandl, M., Wehner, S.: Quantum anonymous transmissions. In: International
  Conference on the Theory and Application of Cryptology and Information
  Security. pp. 217--235. Springer (2005)

\bibitem{cremers2011one}
Cremers, C., Feltz, M.: One-round strongly secure key exchange with perfect
  forward secrecy and deniability. Tech. rep., ETH Zurich (2011)

\bibitem{devetak2005distillation}
Devetak, I., Winter, A.: Distillation of secret key and entanglement from
  quantum states. Proceedings of the Royal Society of London A: Mathematical,
  Physical and Engineering Sciences  \textbf{461}(2053),  207--235 (2005)

\bibitem{di2009new}
Di~Raimondo, M., Gennaro, R.: New approaches for deniable authentication.
  Journal of cryptology  \textbf{22}(4),  572--615 (2009)

\bibitem{di2006deniable}
Di~Raimondo, M., Gennaro, R., Krawczyk, H.: Deniable authentication and key
  exchange. In: Proceedings of the 13th ACM conference on Computer and
  communications security. pp. 400--409. ACM (2006)

\bibitem{dodis2009composability}
Dodis, Y., Katz, J., Smith, A., Walfish, S.: Composability and on-line
  deniability of authentication. In: Theory of Cryptography Conference. pp.
  146--162. Springer (2009)

\bibitem{dwork2004concurrent}
Dwork, C., Naor, M., Sahai, A.: Concurrent zero-knowledge. In: Proceedings of
  the $30^{th}$ Annual ACM Symposium on Theory of Computing. pp. 409--418. STOC
  '98, ACM, New York, NY, USA (1998)

\bibitem{gottesman2002uncloneable}
Gottesman, D.: Uncloneable encryption. Quantum Info. Comput.  \textbf{3}(6),
  581--602 (Nov 2003)

\bibitem{ioannou2011new}
Ioannou, L.M., Mosca, M.: A new spin on quantum cryptography: Avoiding
  trapdoors and embracing public keys. In: International Workshop on
  Post-Quantum Cryptography. pp. 255--274. Springer (2011)

\bibitem{koashi2004monogamy}
Koashi, M., Winter, A.: Monogamy of quantum entanglement and other
  correlations. Physical Review A  \textbf{69}(2),  022309 (2004)

\bibitem{lo1999unconditional}
Lo, H.K., Chau, H.F.: Unconditional security of quantum key distribution over
  arbitrarily long distances. science  \textbf{283}(5410),  2050--2056 (1999)

\bibitem{mayers2001unconditional}
Mayers, D.: Unconditional security in quantum cryptography. Journal of the ACM
  (JACM)  \textbf{48}(3),  351--406 (2001)

\bibitem{mosca2013quantum}
Mosca, M., Stebila, D., Ustao{\u{g}}lu, B.: Quantum key distribution in the
  classical authenticated key exchange framework. In: International Workshop on
  Post-Quantum Cryptography. pp. 136--154. Springer (2013)

\bibitem{nielsen2002quantum}
Nielsen, M.A., Chuang, I.: Quantum computation and quantum information (2002)

\bibitem{pass2003deniability}
Pass, R.: On deniability in the common reference string and random oracle
  model. In: Annual International Cryptology Conference. pp. 316--337. Springer
  (2003)

\bibitem{sheikholeslami2016covert}
Sheikholeslami, A., Bash, B.A., Towsley, D., Goeckel, D., Guha, S.: Covert
  communication over classical-quantum channels. In: Information Theory (ISIT),
  2016 IEEE International Symposium on. pp. 2064--2068. IEEE (2016)

\bibitem{shor2000simple}
Shor, P.W., Preskill, J.: Simple proof of security of the bb84 quantum key
  distribution protocol. Physical review letters  \textbf{85}(2), ~441 (2000)

\bibitem{steane1996multiple}
Steane, A.: Multiple-particle interference and quantum error correction. Proc.
  R. Soc. Lond. A  \textbf{452}(1954),  2551--2577 (1996)

\bibitem{streltsov2012general}
Streltsov, A., Adesso, G., Piani, M., Bru{\ss}, D.: Are general quantum
  correlations monogamous? Physical review letters  \textbf{109}(5),  050503
  (2012)

\bibitem{unger2015deniable}
Unger, N., Goldberg, I.: Deniable key exchanges for secure messaging. In:
  Proceedings of the 22nd acm sigsac conference on computer and communications
  security. pp. 1211--1223. ACM (2015)

\bibitem{wilde2013quantum}
Wilde, M.M.: Quantum information theory. Cambridge University Press (2013)

\bibitem{wootters1982single}
Wootters, W.K., Zurek, W.H.: A single quantum cannot be cloned. Nature
  \textbf{299}(5886),  802--803 (1982)

\end{thebibliography}

\end{document}